\documentclass[twocolumn, pra , footinbib ]{revtex4-2}

\usepackage{graphicx}
\usepackage{amsfonts}
\usepackage{amssymb}
\usepackage{amsthm}
\usepackage{array}
\usepackage{amsmath,mathrsfs}
\usepackage{verbatim} 
\usepackage[colorlinks=true,citecolor=blue,urlcolor=blue]{hyperref}
\usepackage{color}
\usepackage{bbold}
\usepackage{epstopdf}
\usepackage{mathtools}
\usepackage{enumerate}
\usepackage{subcaption}
\usepackage{physics}
\usepackage{subcaption}
\usepackage{soul,xcolor}
\usepackage{newtxtext,newtxmath}
\usepackage{hyphenat}
\usepackage{physics}
\usepackage{braket}

\newtheorem*{proposition*}{Proposition}

\newtheorem{theorem}{Theorem}
\newtheorem*{theorem*}{Theorem}
\newtheorem{corollary}{Corollary}[theorem]
\newtheorem*{corollary*}{Corollary}

\theoremstyle{definition}

\begin{document}
\title{Kirkwood-Dirac Type Quasiprobabilities as Universal Identifiers of Nonclassical Quantum Resources}
\author{Kok Chuan Tan}
\email{bbtankc@gmail.com}
\affiliation{Institute of Fundamental and Frontier Sciences, University of Electronic Science and Technology of China, Chengdu 611731, China}
\affiliation{Key Laboratory of Quantum Physics and Photonic Quantum Information, Ministry of Education, University of Electronic Science and Technology of China, Chengdu 611731, China}
\author{Souradeep Sasmal}
\affiliation{Institute of Fundamental and Frontier Sciences, University of Electronic Science and Technology of China, Chengdu 611731, China}

\affiliation{Key Laboratory of Quantum Physics and Photonic Quantum Information, Ministry of Education, University of Electronic Science and Technology of China, Chengdu 611731, China}


\begin{abstract}

We show that a Kirkwood-Dirac type quasiprobability distribution is sufficient to reveal any arbitrary quantum resource. This is achieved by demonstrating that it is always possible to identify a set of incompatible measurements that distinguishes between   resourceful states and nonresourceful states. The quasiprobability reveals a resourceful quantum state by having at least one quasiprobabilty outcome with a strictly negative numerical value. We also show that there always exists a quasiprobabilty distribution where the total negativity can be interpreted as the geometric distance between a resourceful quantum state to the closest nonresourceful state. It can also be shown that Kirkwood-Dirac type quasiprobability distributions, like the Wigner distribution, can be made informationally complete, in the sense that it can provide complete information about the quantum state while simultaneously revealing nonclassicality whenever a quasiprobability outcome is negative.  Moreover, we demonstrate the existence of sufficiently strong anomalous weak values whenever the quasiprobability distribution is negative, which suggests a means to experimentally test such quasiprobability distributions.  Since incompatible measurements are necessary in order for the quasiprobability to be negative, this result suggests that measurement incompatibility may underlie any quantum advantage gained from utilizing a nonclassical quantum resource. 

\end{abstract}

\maketitle

\section{Introduction}

It is known that quantum theory is deeply incompatible with the classical view of nature \cite{Bell1964,Spekkens2005}. This has attracted intense debates that has not only greatly deepened our understanding of nature \cite{Hardy2001, Home2008, Harrigan2010, Pusey2012, Popescu2014, Bong2020, Patra2023}, and has also led to the development of an immense range of applications that utilizes nonclassical quantum effects in information processing tasks \cite{Ekart1991, Acin2007, Colbeck2012, Brunner2014, Supic2020}. However, not all quantum states can provide superior performance in applications. Most applications rely on the exploiting the nonclassical nature of certain quantum states, which are also called resource states. An example would be entanglement, where entangled quantum states may be used for tasks such as entanglement distillation, quantum teleportation, or harnessing nonlocality \cite{Horodecki2009, Brunner2014}. The challenge is then to identify whether a given quantum state is resourceful or not \cite{Chitambar2019}. For example, quantities such as the entanglement robustness \cite{Vidal1999}, Fischer information \cite{Tan2021}, and quasiprobabilites such as the well known Wigner function \cite{Wigner1932} have been studied for this purpose. In particular, it is known that Wigner function, while providing a signature for nonclassical quantum effects\cite{Tan2020}, is not sufficient to identify quantumness in general\cite{Spekkens2008}.

This brings into question whether the general approach of quasiprobability distributions is robust enough to identify more types of quantum resources. To answer this question, we will consider what we call a Kirkwook-Dirac type quasiprobability distribution.  It will be shown that a Kirkwook-Dirac type quasiprobability is always sufficient to identify any arbitrary quantum resource. More precisely, we establish that for any resourceful quantum state, there exists a quasiprobability distribution of the Kirkwook-Dirac type where at least one of the quasiprobability outcomes yield a strictly negative value. Furthermore, the same distribution always outputs a classical positive probability distribution given any nonresourceful quantum state. We also show that by suitably defining a set of incompatible measurements related to a quantum resource witness, the total negativity of this quasiprobability can be given a geometric interpretation as the Frebenius distance to the closest nonresourceful state. We will also discuss a method of constructing the quasiprobability distribution.   

Note that while negative probabilities are unphysical, this does not prevent quantum quasiprobability distributions from being tested in the laboratory. Recent experiments have demonstrated this through operational quasiprobability distributions involving sequential measurements at two different times \cite{Ryu2019}. To stress the experimental feasibilty of the quasiprobability distribution, we will also establish a connection between the quasiprobability distribution and and anomalous weak values, which has  several applications in information tasks\cite{Dressel2014} and can be mesaured experimentally. The nonclassicality of anomalous weak values has previously been explored  within the contexts of macroscopic realism \cite{Goggin2011, Dressel2011, Pan2020} and non-contextual ontological models \cite{Pussey2014}. Here, we will demonstrate that for every quantum resource with an associated negative quasiprobability outcome, there always exists a corresponding collection of pairs of projectors that lead to anomalous weak values.


\section{Preliminaries} 

\subsection{Negativity and Quasiprobability}
The Kirkwook-Dirac quasiprobability distribution\cite{Kirkwood1933, Dirac1945} is a quasiprobability distribution of the form
\begin{equation}
    p(i,j) = \Tr(\Pi^a_i \Pi^b _j \rho),
\end{equation} where $\Pi^a_i = \ketbra{a_i}$ and $\Pi^b_j = \ketbra{b_j}$. It can be directly verified that as long as $\sum_i \ketbra{a_i} = \sum_j \ketbra{b_j}= \openone$, then $p(i,j)$ sums up to 1. However, it is not a proper probability distribution since $p(i,j)$ can take on negative, or even complex values. If one  only considers the real portions of $p(i,j)$ , then the result is called the Margenau-Hill quasiprobability distribution\cite{Margenau1961}.

For our purposes, we will consider a slightly extended quasiprobability distribution of the Kirkwood-Dirac type, where we allow an arbitrary number of projectors instead of the traditional two. We define a quasiprobability distribution with an ordered set of sets of the type $\{S_0 , \ldots, S_N \}$, where for each $i=0,\ldots,N$, $S_i = \{ \Pi_{i,0}, \ldots , \Pi_{i,d_i} \}$ is a collection of projectors $\Pi_{i,j}$ that satisfy the completeness relation $\sum_{j=0}^{d_i}\Pi_{i,j} =\openone$. The specific ordering of $S_i$ within the set is crucial because projection operators do not commute in general. For a given input state $\rho$, $\{S_0 , \ldots, S_N \}$ completely specifies the quasiprobability distribution as follows 
\begin{equation}
P(x_0, \ldots, x_N \mid \rho) \coloneqq \Tr[\Pi_{0,x_0}\ldots\Pi_{N,x_N} \ \rho ]
\end{equation}
where $x_i = 0 , \ldots, d_i$, $i = 0,\ldots, N$. It is straightforward to verify that this quasiprobability always sums to 1 using the completeness relation. Furthermore, an instance of $(x_0,\ldots,x_N)$ is called a quasiprobability event, and a union of several events forms another event. The quasiprobability of a union of events is given by the the sum of their respective quasiprobabilities, similar to regular probabilities. Therefore, any distribution function resulting from such unions of events is considered a valid quasiprobability distribution.

For the majority of our results, we will consider the special case where $S_i = \{ \Pi_i, \openone - \Pi_i \}$. In this case, each $S_i$ is completely specified by a single projector $\Pi_i$, so we will use the simplified notation $\{ \Pi_0, \ldots, \Pi_N \}$ to represent $\{S_0 , \ldots, S_N \}$ without any ambiguity. The resulting quasiprobability distribution is then given by  
\begin{equation}
    P(x_0,\ldots, x_N \mid \rho ) \coloneqq \Tr[\pi_{x_0}\ldots \pi_{x_N}\rho ]
\end{equation}
where $x_i = 0,1$, $\pi_{x_i=0} \coloneqq \Pi_i$ and $x_i = 0,1$ and $\pi_{x_i=1}\coloneqq \openone - \Pi_i$. Furthermore, for the rest of the discussion, we will always ensure that $P(x_0,\ldots, x_N \mid \rho ) $ are real values.


\subsection{Special case: Wigner distribution} 

The Wigner function \cite{Wigner1932} represents a quantum state via a function on phase space is given by
\begin{equation} \label{wf}
    W(x,p) \coloneqq \int dy \frac{1}{\pi}\mel{x-y}{\rho}{x+y}e^{2  i py }
\end{equation}
assuming $\hbar = 1$. 

Now, by considering $\Pi(x,y,p) \coloneqq \ketbra{\psi(x,y,p)}$, $\ket{\psi(x,y,p)} \coloneqq \sqrt{p}\ket{x+y} +\sqrt{1-p}e^{-2ipy} \ket{x-y}$ and $ p \coloneqq \frac{1}{2} + \sqrt{\frac{1}{4} - \frac{1}{\pi^2}}$, the Wigner function given by Eq.~(\ref{wf}) can be re-expressed as follows
\begin{equation}
    W(x,p)  = \int dy \Tr[\ketbra{x+y}\Pi(x,y,p)\ketbra{x-y}\rho]
\end{equation}

The Wigner distribution $W(x,p)$ is therefore expressed as a sum over $y$ of the quasiprobability $P(x,y,p) \coloneqq \Tr[\ketbra{x+y}\Pi(x,y,p)\ketbra{x-y} \rho ]$ formed by ordered set of projectors $ \{\ketbra{x+y}, \Pi(x,y,p) ,  \ketbra{x-y} \}$. It is worth noting that since the Wigner distribution is informationally complete, meaning one can fully specify the density matrix $\rho$ with information obtained from $W(x,p)$, the quasiprobability distribution $P(x,y,p)$ specified here must also be informationally complete. This is because $W(x,p)$ is retrieved by postprocessing $P(x,y,p)$. The informational completeness of quasiprobabilities is a property that will be further discussed in Theorem~\ref{thm::infoComplete}.


\section{Quasiprobabilities can identify arbitrary quantum resources}

Any nonclassical quantum resource that falls under the resource theory framework is composed of a convex, closed set of quantum states that is considered classical. Any quantum state that lies outside of this closed convex set is called a nonclassical state, or a resourceful state since it is considered a quantum resource. A prominent example of such a quantum resource is entanglement. In entanglement resource theory, the set of classical states is defined as the convex hull of pure product states, while entangled states are defined as quantum states that are not separable.

Recognising the above fact enables us to prove the following general result concerning the relationship between quantum resources and quasiprobabilities. 

\begin{theorem}[Negative quasiprobabilities] \label{thm::pseudoprob}
    Consider any quantum system with corresponding set of classical states $\mathcal{C}$. Then for any quantum state $\rho$ there exists a real valued quasiprobability distribution $P_i(\rho)$ such that $\sum_i P_i(\rho) = 1$, $P_i <0$  for some $i$ if $\rho \notin \mathcal{C}$, and   $P_i \geq 0$  for all $i$ if $\rho \in \mathcal{C}$.
\end{theorem}

\begin{proof}
   Let the Hilbert space be denoted by $\mathcal{H}$. Consider a set of projectors $\{ \Pi_0, \ldots , \Pi_N \}$ acting on the extended Hilbert space $\mathcal{H} \otimes \mathcal{H}_2$ where $\mathcal{H}_2$ is a two-dimensional Hilbert space. Define the function $P(x_0,\ldots, x_N \mid \rho ) \coloneqq \Tr[\pi_{x_0}\ldots \pi_{x_N}\rho \otimes \ketbra{0}]$ where $x_i \in \{0,1\}$, $\pi_{x_i=0} \coloneqq \Pi_i$ and $\pi_{x_i=1}\coloneqq \openone - \Pi_i$. One may directly verify by summing over all strings $(x_0,\ldots, x_N)$ that 
   \begin{equation}
       \sum_{x_0,\ldots, x_N} P(x_0,\ldots, x_N \mid \rho ) = 1
   \end{equation}

   Therefore, $P(x_0,\ldots, x_N \mid \rho )$ forms a quasiprobability distribution.

Now, we show that such a collection of projectors $\{ \Pi_0, \ldots , \Pi_N \}$ always exists. By the hyperplane separation theorem, for any closed convex set $\mathcal{C}$ and any $\rho \notin \mathcal{C}$, there always exists a Hermitian operator $W$ acting on $\mathcal{H}$ such that $\Tr[W \rho] < 0$ while $\Tr[W \sigma] > 0$ for any $\sigma \in \mathcal{C}$. Consider the Hermitian matrix $W' \coloneqq W\otimes\ketbra{0}{0} $ acting on $\mathcal{H}\otimes \mathcal{H}_2$. We observe that $W'$ is a singular matrix since $\Tr[W'\ketbra{\psi}\otimes \ketbra{1}]=0$. It is known that every singular matrix whose operator norm is bounded can be written as a product of projection operators \cite{Oikhberg1999}. Therefore, there must exist some set of projectors $\{ \Pi_0, \ldots , \Pi_N \}$ such that $W' \propto  \Pi_0 \ldots  \Pi_N $, $\Tr[\Pi_0 \ldots  \Pi_N  \rho\otimes \ketbra{0}] < 0$ if $\rho \notin \mathcal{C}$, and  $\Tr[\Pi_0 \ldots  \Pi_N  \sigma \otimes \ketbra{0}] > 0$ for any $\sigma \in \mathcal{C}$.

   Finally, let $P_0$ be the quasiprobability of the event where $(x_0,\ldots, x_N) = (0,\ldots,0)$ while $P_1 $ is the quasiprobability of the event where the strings are not all zeroes. Since $P_0 + P_1 = \sum_{x_0,\ldots, x_N} P(x_0,\ldots, x_N \mid \rho ) = 1$ and $P_0 = \Tr[\Pi_0 \ldots  \Pi_N  \rho\otimes \ketbra{0}] < 0  $, this implies $P_1 = 1- P_0 > 1 > 0$. We have therefore found at least one quasiprobability distribution $P_i = P_i(\rho)$ that satisfies $P_i <0$  for some $i$ given that $\rho \notin \mathcal{C}$, and   $P_i \geq 0$  for all $i$ if $\rho \in \mathcal{C}$.

\end{proof}

In general, the actual quasiprobability distribution described above does not have any particular physical interpretation, beyond the fact that a negative value indicates a resourceful quantum state. However, it is possible to show that one may construct it in such as way as to imbue it with a geometric interpretation, as the following Corollary shows:

\begin{corollary}[Geometric interpretation] \label{cor::Geom}
    It is always possible to construct a quasiprobability distribution where the total negativity is directly proportional to $\min_{\sigma \in\mathcal{C}} \norm{\rho - \sigma}_F$, the Frobenius distance to the closest classical state. 
\end{corollary}

\begin{proof}

For any given $\rho$, let $\sigma_0$ be the closest state in $\mathcal{C}$ with respect to the Frobenius norm, i.e. $\norm{\rho - \sigma_0}_F = \min_{\sigma \in\mathcal{C}} \norm{\rho - \sigma}_F$. Choose 
\begin{equation}
    W = \frac{\sigma_0 - \rho - \Tr[\sigma_0(\sigma_0-\rho)] \openone}{\norm{\rho - \sigma_0}_F}
\end{equation}

Now, considering $W' = W\otimes \ketbra{0}$, it is straightforward to verify by direct evaluation that
\begin{equation}
    \Tr[W'\rho\otimes\ketbra{0}] = - \norm{\rho-\sigma_0}_F = - \min_{\sigma \in\mathcal{C}} \norm{\rho - \sigma}_F
\end{equation}

From the proof of Theorem~\ref{thm::pseudoprob}, it was shown that for any resource witness $W$, satisfying $\Tr[W \rho] < 0$ while $\Tr[W \sigma] > 0$ for any $\sigma \in \mathcal{C}$, then the operator $W' = W\otimes \ketbra{0}$ defined on the extended Hilbert space $\mathcal{H}\otimes \mathcal{H}_2$ is proportional to the product of projectors, given by $W' =  k \Pi_0 \ldots \Pi_N$ where $k \geq 0 $. With such product decomposition $W' =  k \Pi_0 \ldots \Pi_N$, we can always infer the distance to the closest classical state from its negative quasiprobability since
\begin{eqnarray}
\Tr[W'\rho\otimes\ketbra{0}] &=& k \ \Tr[\Pi_0 \ldots \Pi_N \rho\otimes\ketbra{0}] \nonumber \\
&=&  k \ P(x_0 = 0,\ldots, x_N =0   \mid \rho) \nonumber \\
&=& - \min_{\sigma \in\mathcal{C}} \norm{\rho - \sigma}_F
\end{eqnarray}

Furthermore, if we employ the same construction as in the proof of Theorem~\ref{thm::pseudoprob}, then $P_0 = P(x_0 = 0,\ldots, x_N =0   \mid \rho)$ represents the quasiprobability of the event where $(x_0,\ldots, x_N) = (0,\ldots,0)$, while $P_1 $ is the quasiprobability of the event where the strings are not all zeroes. Then $\abs{P_0} = k \min_{\sigma \in\mathcal{C}} \norm{\rho - \sigma}_F $ is the total negativity of the quasiprobability distribution. Therefore, we have provided a construction that satisfies the Corollary.
\end{proof}


\section{Construction of quasiprobability distributions}

The key insight from the preceding discussion is that the crucial property for constructing quasiprobability distributions lies in the ability to map any quantum resource witness $W$ acting on a Hilbert space $\mathcal{H}$ to another witness operator $W'$ acting on the extended Hilbert space $\mathcal{H}\otimes \mathcal{H}_2$. This mapping and extension of the underlying Hilbert space always allows a decomposition as a product of projectors $W' = \Pi_0 \ldots \Pi_N$. 

Now, let's delve into the detailed construction of the set of projectors $\{ \Pi_0, \ldots, \Pi_N\}$ forming the quasiprobability distribution. Consider the simplest possible witness operator $W$ on $\mathcal{H}$ with one positive and one negative eigenvalue, such that $W\ket{k_+} = a\ket{k_+}$, $W\ket{k_-} = -b\ket{k_-}$, and $W =  a \ketbra{k_+} - b \ketbra{k_-}$, where $ 0 < a,b \leq \frac{1}{4}$. By extending the Hilbert space, we map $W \rightarrow W' = W\otimes \ketbra{0}$ acting on $\mathcal{H}\otimes \mathcal{H}_2$. It is evident that $W' \ket{k_\pm}\otimes \ket{1} = 0$, so we are guaranteed that $W'$ has a nontrivial kernel. 

Since $\qty{\ket{k_+} \ket{0}, \ket{k_+} \ket{1}} $ and $\qty{\ket{k_-} \ket{0}, \ket{k_-} \ket{1}} $ each occupy mutually orthogonal subspaces, we can consider the positive and negative portions of the witness operator separately. Let $W'_+ := a \ketbra{k_+} \otimes \ketbra{0}$. It is then possible to show that 
\begin{equation} \label{eqn::wPlus}
    W'_+ = \big( \ketbra{k_+} \otimes \ketbra{0} \big) \ \ketbra{\phi} \ \big( \ketbra{k_+} \otimes \ketbra{0}\big) 
\end{equation} 
where $\ket{\phi} = \sqrt{a}\ket{k_+}\ket{0} + \sqrt{1-a}\ket{k_+}\ket{1}$, which can be verified by a straightforward evaluation.

Similarly, considering $W'_- := -b \ketbra{k_-} \otimes \ketbra{0}$, one can show that 
\begin{eqnarray} \label{eqn::wMinus}
    W'_- &=& \big(\ketbra{k_-} \otimes \ketbra{0}\big) \ \ketbra{\psi_3} \ketbra{\psi_2} \ketbra{\psi_1} \nonumber \\ 
    && \times   \ \big(\ketbra{k_-} \otimes \ketbra{0}\big)
\end{eqnarray}  
where $\ket{\psi_1} = \sqrt{1-\lambda}\ket{k_-}\ket{1}- \sqrt{\lambda}\ket{k_-}\ket{0}, $ $\ket{\psi_2} = \ket{k_-}\ket{1}, $ $\ket{\psi_3} = \sqrt{1-\lambda}\ket{k_-}\ket{1}+ \sqrt{\lambda}\ket{k_-}\ket{0}, $ and $\lambda = \frac{1 - \sqrt{1-4b}}{2}$. Note that $W'_+$ and $W'_-$ are both expressed as a product of projectors such that $W'_+ = \Pi_{+,0} \Pi_{+,1} \Pi_{+,0}$ and $W'_- = \Pi_{-,0} \Pi_{-,3} \Pi_{-,2} \Pi_{-,1} \Pi_{-,0}$ . 

From here, we exploit the property that $W'_+$ and $W'_-$ and the corresponding projectors that form each of them act on orthogonal subspaces to write $W' = W'_+ + W'_- =    \Pi_0 \Pi_3 \Pi_2 \Pi_1 \Pi_0$ where $\Pi_0 = \ketbra{k_-} \otimes  \ketbra{0} +\ketbra{k_+} \otimes \ketbra{0} $, $\Pi_1 = \Pi_{+,1} + \Pi_{-,1}$, $\Pi_2 = \Pi_{+,1} + \Pi_{-,2}$ and $\Pi_3 = \Pi_{+,1} + \Pi_{-,2}$. It can easily be verified that each projector satisfies $\Pi_i^2 = \Pi_i \ \forall i \in \{0,1,2,3\}$. We have therefore constructed a set of projectors $\{ \Pi_0,\Pi_1, \Pi_2, \Pi_3, \Pi_0 \}$ on the extended Hilbert space that forms the required quasiprobability distribution. 

So far, the construction has been showcased for a simple $W$ with one positive and one negative eigenvalue. Remarkably, this serves as a sufficient basis to establish the general case. The reasoning behind this is that even for more general witness operators $W$ featuring multiple positive and negative eigenvalues, to assign each of these eigenvalues a product of projectors, acting within its own subspace. This is feasible because of the orthogonality between each of these subspaces. Consequently, the same line of reasoning holds, making the approach adaptable and extendable for the general witness operator $W$. 


\section{Relationship with anomalous weak values} 

A quantum weak value with respect to a Hermitian operator $A$ is defined \cite{Aharonov1988} as $A_w \coloneqq\braket{\phi\mid A \mid \psi}/\braket{\phi \mid \psi}$ for a quantum state initially prepared in the state $\ket{\psi}$ and a subsequent postselected state $\ket{\phi}$. Most notably,  even if $A$ is Hermitian, $A_w$ may not always be real-valued and that the imaginary and real portions of the weak value manifests itself differently\cite{Jozsa2007, Bartlett2012, Wagner2021}. We say that a weak value $A_w$ as anomalous whenever its real portion $Re(A_w)$ is either smaller than the smallest eigenvalue of $A$ or larger than the largest eigenvalue of $A$. Considering the spectral decomposition, $A=\sum_a a  \pi^a$ and incorporating it into the definition of weak value, it becomes evident that $Re(A_w)=\sum_a Re(\pi_w^a)$. 

An anomalous weak value for any observable therefore implies anomalous weak value for projectors. Given that $\sum_a \pi^a_w = \openone$, if one projector has $Re(\pi^a_w)>1$, then another must have $Re(\pi^{a'}_w)<0$. without loss of generality, we can always assign the anomalous weak value to the projector $\pi$ having $Re(\pi_w)<0$. The following theorem relates the anomalous weak value to the negativity of the quasiprobability distribution.

\begin{theorem}[Anomalous weak values for resourceful states] \label{thm::anomalousWeak}
Consider any quantum system with a corresponding set of classical states $\mathcal{C}$. Then, for every $\rho \notin \mathcal{C}$, there must exist a collection of pairs of projectors $\qty{\qty( \pi_1(j), \pi_2(j))}_j$ and a conical combination (sum over positive coefficients) of weak values $\mathcal{W}(\rho) = \sum_j k(j) \braket{\qty(\pi_2(j))_w}^{\pi_1(j)}_{\rho\otimes \ketbra{0}}$, $k(j) \geq 0$, that satisfies $\mathcal{W}(\rho) <0$ and $\mathcal{W}(\sigma) \geq 0$ for every $\sigma \in \mathcal{C}$. 
    
Furthermore, there must be at least one $j$ such that the weak value $\expval{(\pi_2(j))_w}^{\pi_1(j)}_{\rho\otimes \ketbra{0}}$ is anomalous and has a negative real part. The above-mentioned conical combination of weak values is, therefore, negative only if the anomalous weak values are sufficiently negative. 
\end{theorem}

\begin{proof}
From the preceding discussions, it is evident that a resource witness $W$ can always be mapped to another operator $W' = W\otimes \ketbra{0}$ acting on an extended Hilbert space $\mathcal{H}\otimes \mathcal{H}_2$. Note that $W'$ can be expressed as a sum of positive and negative contributions. For simplicity, let us initially consider the case $W' = W'_+ + W'_-$, which comprises one positive ($W'_+$) and one negative ($W_-'$) eigenvalue contribution. 

Firstly, we focus on the positive contribution. By definition, we have $W'_+ = a \ketbra{k_+}\otimes \ketbra{0}$ with $a > 0 $. This can be trivially re-formulated as $\Tr[W'_+ \ \rho] = a \Tr[\ketbra{k_+} \rho ] \langle \ketbra{k_+}\otimes \ketbra{0}_W \rangle_{\rho\otimes \ketbra{0}}^{\ketbra{k_+}\otimes \ketbra{0}}$. The expectation value of $W'_+$ is thus expressed in terms of a weak value, and in this case, the weak value is non-negative. Note that $\pi_1(0) = \pi_2(0) = \ketbra{k_+}\otimes \ketbra{0} $ are projectors.

Next, we shift our attention to the negative contribution. By definition, $W'_- = -b \ketbra{k_-}\otimes \ketbra{0}$, where $b> 0$. Considering the expectation value $\Tr[W'_- \rho \otimes \ketbra{0}]$, we define $\ket{\psi_1} \coloneqq \frac{1}{\sqrt{2}}(\ket{0} - \ket{1})$ and $\ket{\psi_2} \coloneqq (\sqrt{1-p}\ket{0} + \sqrt{p}\ket{0})$. One may verify that 
    \begin{equation} \label{nwv}
        \ket{0}\braket{0 | \psi_1}\braket{\psi_1 | \psi_2} \braket{\psi_2 | 0} \bra{0} = \frac{1-p - \sqrt{p(1-p)}}{2} \ketbra{0}
    \end{equation}
From the above Eq.~(\ref{nwv}), the coefficient $ (1-p - \sqrt{p(1-p)})$ is negative when $1/2 < p < 1$. By choosing the value of the parameter $p$ within this range, we obtain
\begin{equation} \label{nwv1}
    -\ketbra{0} =  c \ \ket{0}\braket{0 | \psi_1}\braket{\psi_1 | \psi_2} \braket{\psi_2 | 0} \bra{0}
\end{equation}
where $c = 2/(\sqrt{p(1-p)}-1+p) > 0$. Now, substituting this into the expression for $W_-'$, we get
\begin{equation}
W_-' = bc\ketbra{k_-} \otimes \ket{0}\braket{0 | \psi_1}\braket{\psi_1 | \psi_2} \braket{\psi_2 | 0} \bra{0}
\end{equation} 
where $b,c > 0$. The expectation value of $W_-'$ is evaluated as follows
\begin{eqnarray}
&& \Tr[W_-'\rho \otimes \ketbra{0}] \nonumber \\
&=& bc \Tr[ \ketbra{k_-} \otimes \ket{0}\braket{0 | \psi_1}\braket{\psi_1|\psi_2} \braket{\psi_2|0} \bra{0} \rho \otimes \ketbra{0}] \nonumber \\
&=& bc \text{Tr} \big[\qty(\openone \otimes \ketbra{0})\ketbra{k_-} \otimes \ket{\psi_1} \braket{\psi_1|\psi_2} \bra{\psi_2}  \nonumber \\
&& \times \qty(\openone \otimes \ketbra{0}) \ \rho \otimes \ketbra{0} \big] \nonumber \\
&=& bc \Tr[\ketbra{k_-} \otimes  \ket{\psi_1}\braket{\psi_1|\psi_2} \bra{\psi_2} \  \rho \otimes \ketbra{0}] \nonumber \\
&=& d  \ \Big\langle \big(\ketbra{k_-} \otimes \ketbra{\psi_2}\big)_w \Big\rangle^{\ketbra{k_-} \otimes  \ketbra{\psi_1}}_{\rho \otimes \ketbra{0}}
\end{eqnarray} 
where $d \coloneqq bc \Tr[\ketbra{k_-} \otimes  \ketbra{\psi_1} \rho \otimes \ketbra{0}] >0$. In the final line, the expectation value of $W_-'$ is expressed as a weak value. Furthermore, it is the weak value of the projection operator $\ketbra{k_-} \otimes \ketbra{\psi_2}$, whose spectrum is necessarily non-negative. If $\rho \notin \mathcal{C}$, then $\Tr[W_-'\rho \otimes \ketbra{0}]  < 0$, which is clearly negative and implies that the corresponding weak value is also negative and hence, anomalous. Note that $\pi_1(1) = \ketbra{k_-} \otimes \ketbra{\psi_2} $  and $\pi_1(1) = \ketbra{k_-} \otimes  \ketbra{\psi_1}$ are projectors.

Thus, the above argument infers that the expectation value of $W_-'$ can indeed be expressed as the intended conical combination of weak values as follows
\begin{eqnarray}
\Tr[W'\rho\otimes \ketbra{0}] &=& k_0 \ \Big\langle \pi_1(0)_w \Big\rangle^{\pi_2(0)}_{\rho \otimes \ketbra{0}} \nonumber \\
&& + k_1 \ \Big\langle \pi_1(1)_w  \Big\rangle^{\pi_2(1)}_{\rho \otimes \ketbra{0}} \ \ ;
\end{eqnarray}
where $k_0, k_1 >0$ are positive real numbers. Hence, the Theorem stands proven for the case where $W' = W'_+ + W'_-$ has exactly one positive and one negative eigenvalue contribution. Importantly, this argument can be straightforwardly extended to any number of positive and negative contributions, demonstrating the proof for the general case.
\end{proof}


\section{Informational completeness and resource identification} 

An informationally complete quasiprobability is one that contains complete information about a quantum state. It can be shown that for any quantum system with a well defined notion of classical states $\mathcal{C}$, one can formulate a quasiprobability distribution that identifies resourceful states, while simultaneously containing complete information about the state.

\begin{theorem} [Informationally complete quasiprobability] \label{thm::infoComplete}
    For any quantum system with corresponding set of classical states $\mathcal{C}$, there always exists quasiprobability distribution $P(y \mid \rho )$ that is an informationally complete representation of $\rho$, and is negative if and only if $\rho \notin C$. 
\end{theorem}

\begin{proof}
Let $\{\Pi_0, \ldots, \Pi_d \}$ be any set of informational complete, rank 1 projectors (such as SIC-POVMs). For any given $\rho$, let $\{\Pi_{d+1}, \ldots, \Pi_N \}$ be the set of projectors satisfying $\Tr[\Pi_{d+1}\ldots \Pi_N \rho \otimes \ketbra{0}] < 0 \ \forall \rho \notin \mathcal{C}$ and  $\Tr[\Pi_{d+1}\ldots \Pi_N \sigma \otimes \ketbra{0}] \geq 0 \  \forall \sigma\in \mathcal{C}$. The existence of such a set of projectors has already been shown in Theorem~\ref{thm::pseudoprob}. We observe that the concatenation of these two ordered sets, $\{ \Pi_0, \ldots, \Pi_N \}$ forms another valid quasiprobability distribution.

Define $P(x_0,\ldots, x_N \mid \rho ) \coloneqq \Tr[\pi_{x_0}\ldots \pi_{x_N}\rho \otimes \ketbra{0}]$, where $x_i \in \{0,1\}$ and $\pi_{x_i=0} \coloneqq \Pi_i$ and $\pi_{x_i=1}\coloneqq \openone - \Pi_i$. For any $0 \leq k \leq d$, consider the marginal 
\begin{eqnarray}
P(x_k = 0  \mid \rho ) &=& \sum_{\substack{ x_i \\
    0 \leq i \leq N \\
    i \neq k}}P(x_0,\ldots, x_k = 0 , \ldots  x_N \mid \rho ) \nonumber \\ 
    &=& \Tr[\Pi_k \ \rho \otimes \ketbra{0}] \geq 0
    \end{eqnarray} 
    Since $\{\Pi_0, \ldots, \Pi_d \}$ is informationally complete, it is clear that  the quasiprobability $P(x_0,\ldots, x_N \mid \rho )$ is also informationally complete. 
    
Furthermore, we have  
 \begin{align}
        &P(x_d+1 = 0 , \ldots, x_N = 0   \mid \rho )  \nonumber \\
        &\qquad =  \sum_{\substack{ x_i \\
    0 \leq i \leq d \\
    i \neq k}}P(x_0,\ldots, x_d, x_d+1 = 0 , \ldots, x_N = 0 \mid \rho ) \nonumber \\ 
    &\qquad =  \Tr[\Pi_{d+1} \ldots \Pi_{N} \ \rho \otimes \ketbra{0}]
    \end{align} 
which is negative only if $\rho \notin \mathcal{C}$.

Therefore, from $\{\Pi_0, \ldots, \Pi_N \}$, we have as marginals, the quasiprobability $P(x_k = 0  \mid \rho ),$ $0\leq k \leq d$ which is informationally complete and always non-negative, as well as $P(x_{d+1} = 0 , \ldots, x_N = 0   \mid \rho ),$ which is negative only when $\rho$ is a resourceful state. Thus, we can define the following
    \begin{eqnarray}
    P(y = k \mid \rho ) &\coloneqq& P(x_k = 0  \mid \rho )\;,\; k= 0, \ldots, d \nonumber \\
        P(y= d+1 \mid \rho )   &\coloneqq& P(x_{d+1} = 0 , \ldots, x_N = 0   \mid \rho ) \nonumber \\
        P(y= d+2 \mid \rho ) &\coloneqq& 1- \sum_{y-0}^{d+1} P(y \mid \rho) \nonumber
    \end{eqnarray} 
where each of the terms $P(y \mid \rho)$ is a sum of quasiprobability events taken from the underlying quasiprobability distribution $P(x_0,\ldots, x_N \mid \rho )$, so $P(y \mid \rho)$ is again a valid quasiprobability. They are clearly all real number and sum to 1. $P(y \mid \rho )$ is therefore  an informationally complete representation of the state $\rho$, consequently, is also negative \textit{iff} $\rho \notin \mathcal{C}$.
\end{proof}


\section{A qubit example}

We consider the simplest example on a qubit to illustrate our key findings. Consider a single qubit system where the quantum resource of interest is quantum coherence, and the set of classical states is defined as $\mathcal{C} \coloneqq \{ \sigma = (\openone + r \hat{z}\cdot \Vec{\sigma})/2 : \abs{r} \leq 1 \}$. The set $\mathcal{C}$ can be visualized as the set of qubit states lying on the $z$-axis in the Bloch sphere representation of the qubit. In this context, any qubit state not lying on the $z$-axis is said to have quantum coherence, designating it as a resourceful state. A qubit enables us to  visualise the geometric interpretation of the quasiprobability as outlined in Corollary~\ref{cor::Geom}.

For the construction of the quasiprobability, we will focus on a pure qubit state $\ket{\psi}$. Recall that every pure qubit state is expressed as $\ket{\psi} = \cos(\theta/2)\ket{0}+ e^{i\phi} \sin(\theta/2)\ket{1}$ where $\theta$ is the polar angle and $\phi$ is the azimuthal angle in the Bloch sphere representation. Due to rotational symmetry about the $z$-axis in the system, it suffices to set $\phi = 0$, so we will consider only $\ket{\psi} = \cos(\theta/2)\ket{0}+ \sin(\theta/2)\ket{1}$. The corresponding density operator $\rho = \ketbra{\psi}$ and consider the following resource witness 
\begin{equation}
    W = \frac{\sigma_0 - \rho - \Tr[\sigma_0(\sigma_0-\rho)] \openone}{\norm{\rho - \sigma_0}_F}
\end{equation}
where $\sigma_0$ is the closest state in the set $\mathcal{C}$ to $\rho$. The state $\ket{\psi} = \cos(\theta/2)\ket{0}+ \sin(\theta/2)\ket{1}$ has Cartesian coordinates $(x,y,z) = (\sin{\theta}, 0 ,\cos(\theta))$ in the Bloch sphere representation, therefore, the closest quantum state on the $z$-axis has Cartesian coordinates $(0,0,\cos(\theta))$, which corresponds to the density matrix $\sigma_0 = (\openone + \cos(\theta))/2$. Substituting in the expressions for $\sigma_0$ and $\rho$, we obtain $W = -\sigma_x$.

Following Theorem~\ref{thm::pseudoprob}, we extend the Hilbert space by mapping $W$ to $W' = -\frac{1}{4} W \otimes \ketbra{0}$. Let $W'_+ \coloneqq \frac{1}{4} \ketbra{-} \otimes \ketbra{0}$ and $W'_- \coloneqq -\frac{1}{4}\ketbra{+}\otimes \ketbra{0}$, such that $W' = W'_+ + W'_-$. We then define the following projection operators 
\begin{eqnarray}
    \Pi_{+,0} &\coloneqq& \ketbra{-}\otimes \ketbra{0} \\ 
    \Pi_{+,1} &\coloneqq& \qty(\frac{1}{2}\ket{-}\ket{0} + \sqrt{\frac{3}{4}} \ket{-}\ket{1} ) \cdot (h.c.)
\end{eqnarray} 
and
\begin{eqnarray}
    \Pi_{-,0} &\coloneqq&  \ketbra{+}\otimes \ketbra{0} \\ 
    \Pi_{-,1} &\coloneqq& \qty(\sqrt{\frac{1}{2}}\ket{+}\ket{0} + \sqrt{\frac{1}{2}} \ket{+}\ket{1} ) \cdot (h.c.) \\
    \Pi_{-,2} &\coloneqq& \ketbra{+}\otimes \ketbra{1} \\
    \Pi_{-,3} &\coloneqq& \qty(\sqrt{\frac{1}{2}}\ket{+}\ket{0} - \sqrt{\frac{1}{2}} \ket{+}\ket{1} ) \cdot (h.c.)
\end{eqnarray}

It may be directly verified using the above definitions that $W'= \Pi_0 \Pi_3 \Pi_2 \Pi_1 \Pi_0$ where $\Pi_0 = \openone \otimes \ketbra{0} $, $\Pi_1 = \Pi_{+,1} + \Pi_{-,1}$, $\Pi_2 = \Pi_{+,1} + \Pi_{-,2}$ and $\Pi_3 = \Pi_{+,1} + \Pi_{-,2}$. Additionally, it holds that $\Tr[ \Pi_0 \Pi_3 \Pi_2 \Pi_1 \Pi_0 \rho] < 0$ if $\rho \notin \mathcal{C}$ and $\Tr[ \Pi_0 \Pi_3 \Pi_2 \Pi_1 \Pi_0 \rho]$ for every $\sigma \in \mathcal{C}$, consistent with the statement in Theorem~\ref{thm::pseudoprob}.

Next, we evaluate the expectation value of $W'_+$
\begin{equation}
    \Tr[W'_+\rho\otimes \ketbra{0} ] = k_0 \ \Big\langle\big(\pi_2(0)\big)_w\Big\rangle^{\pi_1(0)}_{\rho_\otimes \ketbra{0}},
\end{equation} 
where $k_0 = \frac{1}{4} \Tr[ \ketbra{-}  \rho]$ and $\pi_1(0) = \pi_2(0) = \ketbra{-} \otimes \ketbra{0}$.

Likewise, the expectation value of $W'_-$ is evaluated as
\begin{eqnarray}
    \Tr[W_-\rho\otimes \ketbra{0}] &=& \frac{c}{4} \Tr[\ketbra{+} \rho \otimes \ketbra{-} \ketbra{\psi} \ketbra{0}] \nonumber \\
    &=& k_1 \ \Big\langle\big(\pi_2(1)\big)_w\Big\rangle^{\pi_1(1)}_{\rho \otimes \ketbra{0}}
\end{eqnarray} 
where $k_1 \coloneqq c\Tr[\ketbra{+} \otimes \ketbra{-} \rho \otimes \ketbra{0} ] /4$, $\pi_1(1) \coloneqq \ketbra{+} \otimes \ketbra{-}$, $\pi_2(1) \coloneqq \ketbra{+} \otimes \ketbra{\psi}$. Consequently, it follows that $\Tr[W'\rho \otimes \ketbra{0}] = k_0 \expval{(\pi_2(0))_W}^{\pi_1(0)}_{\rho_\otimes \ketbra{0}} + k_1 \expval{(\pi_2(1))_W}^{\pi_1(1)}_{\rho \otimes \ketbra{0}}$, which represents a conical combination of weak values. Choosing $\rho = \ketbra{+}$ which is a resourceful state, we obtain the value $\expval{(\pi_2(1))_W}^{\pi_1(1)}_{\rho \otimes \ketbra{0}} = (1-\sqrt{3})/(2\sqrt{2}) < 0 $, thus confirming that the weak value is both negative and anomalous. This  is in line with the statement from Theorem~\ref{thm::anomalousWeak}.


\section{Discussion and Conclusion} 

Quantum quasiprobabilities have traditionally been valued for their ability to provide key signs of a breakdown of classical physics and the emergence of specific quantum effects. The present study considers the limits of the quantum quasiprobabilty approach by studying their ability to identify more general notions of nonclassicality. By employing Kirkwood-Dirac type quasiprobability distributions, we were able to show that this type of quasiprobability distribution is already able to identify arbitrary notions of quantum resources. In particular, the Kirkwood-Dirac type quasiprobabilities we considered violate Kolmogorov's positivity axiom, when given a resourceful quantum state as input. For nonresourceful quantum states, it always outputs a positive probability distribution. A common property of several quasiprobabilty distributions, such as the Wigner distribution, the Glauber-Sudarshan P function\cite{Glauber1963,Sudarshan1963} and Husimi's Q function \cite{Husimi1940, Tan2019}, is that they contain full information about the quantum state. We have shown that Kirkwood-Dirac type quasiprobability distributions can also be made informationally complete, like the Wigner function, without compromising its ability to identify arbitrary resourceful states. The total negativity can also be imbued with a direct geometric interpretation, also without compromising its other qualities.

We also related every negative quasiprobability probabilty outcome to sufficiently strong anomalous weak values. Thus, the emergence of strongly anomalous weak values coincides with the emergence of negative quasiprobabilities in our approach. This is in line with existing results \cite{Spekkens2008, Pussey2014} that demonstrate the conflict of negativity (quasiprobability) and anomalous weak values with non-contextual ontological models, suggesting that quasiprobability and anomalous weak values closely related equivalent quantum features. It is unclear at present whether there is a more precise quantitative relationship between negative quasiprobabilities and the onset of anomalous weak values. We leave this for future work.

Finally, our quasiprobability approach for identifying quantum resource brings out another crucial feature: quantum resources depend not only on the quantum state but also on the type of measurements performed on the state. In order to yield a negative quasiprobability using Kirkwood-Dirac type quasidistributions, incompatible measurement projectors are necessary. One possible interpretation of this is that all quantum resources derive their nonclassical features from measurement incompatibility. It is interesting to also consider whether this can be extended to a more general class of measurements such as noisy projectors (unbiased POVM) \cite{Busch1996book} and and whether such can also have a meaningful interpretation. Our results also suggests possible new avenues for the use of quasiprobability distributions in applications such as quantum metrology. Separately, it has been shown that there always exists a parameter estimation task where an arbitrary quantum resource can provide an advantage\cite{Tan2021}, and here, we have shown that arbitrary quantum resource must have negative quasiprobabilities which requires measurement incompatibility. This is in line with recent studies that have shown that incompatible measurements and contextuality provide a quantum advantage in metrological tasks \cite{jae2023a, jae2023b}.  

In sum, the result that every quantum resource corresponds to an informationally complete quasiprobability distribution and gives rise to anomalous weak values significantly broadens the scope for corroborating non-classicality and its advantages in various quantum theory-based applications. Simultaneously, it revives the foundational question \cite{Pan2020, Catani2022, Catani2023a, Catani2023b} of whether different forms of non-classical phenomena, such as contextuality, nonlocality, indefinite-causal order, violations of macro-realism through Leggett-Garg inequality and interference phenomena stem from a more fundamental and unique non-classical feature or not.


\acknowledgments 

K.C.T. and S.S. acknowledges support by the National Natural Science Fund of China (Grant No. G0512250610191)


\bibliography{quasiprobabilities}

\end{document}